\documentclass[12pt,twoside,openright,a4paper]{report}

\usepackage{setspace}
\usepackage{amsmath}
\usepackage{amsthm}
\usepackage{mathrsfs}
\usepackage{amssymb}
\usepackage{psfrag}
\usepackage{graphicx,colordvi,psfrag}
\usepackage{epstopdf}
\usepackage{epsfig}
\usepackage{calc,pstricks, pgf, xcolor}
\usepackage{lscape}
\usepackage[colorlinks=true]{hyperref}

\newtheorem{lemma}{Lemma}
\theoremstyle{plain} \newtheorem{theorem}{Theorem}
\theoremstyle{plain} \newtheorem*{untheorem}{Theorem}
\theoremstyle{definition} \newtheorem{definition}{Definition}
\theoremstyle{definition} 
\theoremstyle{remark} \newtheorem{remark}{Remark}[chapter]

\newcommand{\E}{\mathbb{E}}

\newcommand{\tv}{\boldsymbol{\theta}}
\newcommand{\gv}{\boldsymbol{\gamma}}
\newcommand{\Lkappa}{\mathcal{K}}
\newcommand{\htv}{\hat{\tv}}
\newcommand{\XX}{|\mathcal{X}|}
\newcommand{\YY}{|\mathcal{Y}|}
\newcommand{\XY}{|\mathcal{X}||\mathcal{Y}|}
\newcommand{\hXY}{\XY/2}
\newcommand{\BXY}{B}
\newcommand{\condvar}{\frac{\theta P}{\theta + P}}
\newcommand{\Wiener}{\frac{P}{\theta + P}}
\newcommand{\ph}{\hat{\pi}}

\begin{document}

\pagestyle{empty}

\newcommand{\thesistitle}{Rateless Codes for Finite Message Set}
\newcommand{\titlespacing}{24 mm}

\begin{center}
 \vspace*{-12 mm}
 \includegraphics[scale=1,angle=0]{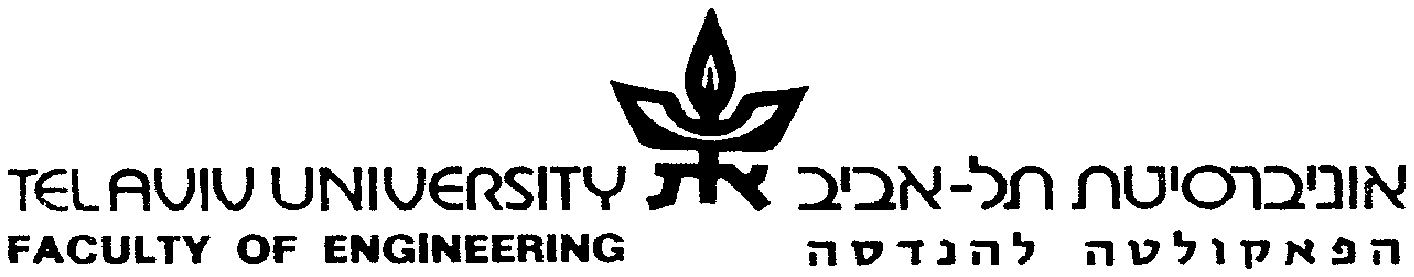}
 \vspace{-4 mm} \large Department of Electrical  Engineering  - Systems\\
 \vspace{\titlespacing}
\vspace{15mm}
 \Huge  {\bf \thesistitle} \\
\vspace{5mm}
 \vspace{\titlespacing}
  \large Thesis submitted toward the degree of \\``Master of Science in Electrical and Electronic Engineering''\\
 \vspace{8 mm} by \\
 \vspace{5 mm} \Large {\bf Navot Blits} \\
\vspace{10mm}
\vspace{-10mm} \vspace{-20mm}
 \vspace{67 mm} \large 2012 \\
\end{center}

\cleardoublepage
 \begin{center}
 \vspace*{-12 mm}
 \includegraphics[scale=1,angle=0]{tau_logo.eps}
 \vspace{-4 mm} \large Department of Electrical  Engineering  - Systems\\
 \vspace{\titlespacing}
 \Huge  {\bf \thesistitle} \\
 \vspace{\titlespacing}
  \large Thesis submitted toward the degree of \\``Master of Science in Electrical and Electronic Engineering''\\
 \vspace{8 mm} by \\
 \vspace{5 mm} \Large {\bf Navot Blits} \\
 \vspace{8 mm}\large This research was conducted at the\\
 \vspace{4 mm} Department of Electrical Engineering - Systems,\\
 \vspace{4 mm} Tel Aviv University\\
 \vspace{4 mm} under the supervision of\\
 \vspace{6 mm} \Large {\bf Prof. Meir Feder} \\
 \vspace{7.5 mm} \large 2012 \\

\end{center}
\large

\normalsize

\cleardoublepage

\begin{center}
~\\
 \vspace{36mm} \large This study was carried out under the supervision of\\
 \vspace{6mm} \Large {\bf Prof. Meir Feder}
\end{center}

\cleardoublepage

\baselineskip0.7cm

\chapter*{Acknowledgements}
I would like to thank my supervisor, Prof. Meir Feder, for his devoted mentoring throughout my research. Without Meir's enthusiasm, curiosity and endless support, this thesis could not have been what it is now. His depth as an information theorist, his broad-mindedness as a scientist and his openheartedness inspired me all along the way. I am grateful for having given the opportunity to walk this path with him.

I would also like to thank Or Ordentlich and Yuval Lomnitz for their interest and insights, which helped me tackle some of the most challenging problems in my research.

\cleardoublepage

\begin{abstract}
In this study we consider rateless coding over discrete memoryless channels (DMC) with feedback. Unlike traditional fixed-rate codes, in rateless codes each codeword in infinitely long, and the decoding time depends on the confidence level of the decoder. Using rateless codes along with sequential decoding, and allowing a fixed probability of error at the decoder, we obtain results for several communication scenarios. The results shown here are non-asymptotic, in the sense that the size of the message set is finite.

First we consider the transmission of equiprobable messages using rateless codes over a DMC, where the decoder knows the channel law. We obtain an achievable rate for a fixed error probability and a finite message set. We show that as the message set size grows, the achievable rate approaches the optimum rate for this setting. We then consider the \emph{universal} case, in which the channel law is unknown to the decoder. We introduce a novel decoder that uses a mixture probability assignment instead of the unknown channel law, and obtain an achievable rate for this case.

Finally, we extend the scope for more advanced settings. We use different flavors of the rateless coding scheme for joint source-channel coding, coding with side-information and a combination of the two with universal coding, which yields a communication scheme that does not require any information on the source, the channel, or the amount the side information at the receiver.
\end{abstract}

\cleardoublepage

\pagestyle{plain}
\pagenumbering{roman}

\tableofcontents
\singlespace
\listoffigures

\cleardoublepage

\baselineskip0.7cm

\pagestyle{plain}
\pagenumbering{arabic}

\chapter{Introduction}
\section{Background} \label{sec:Background}
In traditional channel coding schemes the code rate, which is the ratio between the lengths of the encoder's input and output blocks, is an integral part of the code definition. If one of $M$ messages is to be encoded at rate $R$, then the corresponding codeword has length $n=(\log M)/R$. Provided that the rate is chosen properly, the error probability decreases as $M$ grows. The capacity of the channel $C$ is defined as the largest value of $R$ for which the error probability can vanish.

An alternative approach to fixed-rate channel coding is \emph{rateless codes}. In this approach, we abandon the basic assumption of a fixed coding rate, and allow the codeword length, and hence also the rate, to depend on the channel conditions. When the encoder wants to send a certain message, it starts transmitting symbols from an infinite-length codeword. The decoder receives the symbols that passed through the channel and when it is confident enough about the message, it makes a decision. Perhaps the simplest example of a rateless code is the following (see e.g. \cite[Ch.3]{Nadav} or \cite[Ch.7]{Cover}). Suppose that we have a binary erasure channel (BEC) with erasure probability $\delta$. Suppose also that noiseless feedback exists, i.e. the encoder at time instant $n$ has an access to the outputs of the channel at times $1,\ldots,n-1$. We use a simple repetition coding, in which each binary symbol is retransmitted until the decoder receives an unerased symbol. Since the erasure probability is $\delta$, the expected number of transmissions until an unerased symbol is received is $1/(1-\delta)$. This transmission time implies a rate of $1-\delta$, which is exactly the capacity of the binary erasure channel. This simple setting exemplifies some important concepts of rateless codes. First, the transmission time is not fixed, but rather is a random variable (geometrically-distributed in the above case); second, when the length of the transmission is set dynamically, the error probability may be controllable. In this case the transmission is only terminated once the decoder \emph{knows} what message has been transmitted, so the error probability of this coding scheme is zero; third, the code design is rate-independent. In fact, this code can be used for any binary erasure channel; fourth, the continuity of the transmission requires feedback to the encoder. Indeed, as we shall see in this thesis, when rateless codes are used for point-to-point communication, some form of feedback, which can be limited to decision feedback, must exist to enable continuity. However, rateless codes are also invaluable for other settings such as multicast or broadcast communications, in which the existence of feedback is not explicitly required. Shulman \cite{Nadav} introduced the concept of \emph{Static Broadcasting}, in which the transmitter sends a message to multiple users, and each user remains connected until it retrieved enough symbols to make a confident decision. This scheme does not require feedback; the user remains online only as much as it needs, and the rate is determined according to the time the user spent online.

In this thesis we assume a discrete memoryless channel (DMC) with feedback, and devise rateless coding schemes which allow a small (but fixed) error probability $\epsilon$. We investigate the dependence between the rate, the error probability and the size of the message set. The entire analysis is done for a finite message set, and we show that when the size of the message set is taken to infinity, our results agree with classic results from coding theory. We also investigate the rate of convergence to these results. We start by building a simple rateless coding scheme for a known channel. The motivation for this method is due to Wald's analysis (see \cite[Ch.3]{Wald}), where he demonstrated that the Sequential Probability Ratio Test (SPRT) performs like the most powerful test in terms of error probabilities, while using about half the samples on average.

Building on the rateless coding scheme devised for the case of known channel, we obtain a \emph{universal} channel coding scheme that does not require channel knowledge at the receiver. Unlike previous results on universal decoding, the results here are non-asymptotic and are valid for an arbitrary message set size. We then extend the coding scheme to joint source-channel coding, and show that optimal rate is achievable even when the encoder is uninformed on the source statistics. Next, we use a rateless coding scheme for source coding with side information at the receiver and show that the Slepian-Wolf rate for this scenario is achievable even when the encoder is unaware of the amount of side information. Finally, we show how to combine the above-mentioned techniques with universal source coding, to obtain a scheme that can operate when the statistics of both the source and the channel are unknown, potentially using side information that is obscure to the encoder.

Our work follows previous results discovered by Shulman \cite{Nadav} for the universal case, where the decoder is ignorant of the channel law. In particular, a sequential version of the maximal mutual information (MMI) decoder \cite{CsiszarKorner} is used for universal channel decoding and joint-source channel coding, including the case of side information at the decoder. However, the results in \cite{Nadav} are asymptotic in the size of the message set, while the analysis here is made for a fixed size of the message set. For the case of known channel, the decoder used here can be viewed as the counterpart of the sequential MMI decoder that uses the channel law rather than the empirical mutual information. This scheme has been originally introduced by Polyanskiy in \cite{PPV}, where it is proven to achieve the best variable-length coding rate. While the analysis in \cite{PPV} concentrates on finding the best achievable size of the message set with a constraint on the average decoding time, in this paper we seek the optimum decoding time for a fixed size of the message set. More importantly, the analysis introduced here is then extended naturally to apply for the case of unknown channel, where we use a novel universal decoder, as well as for joint source-channel coding with and without side information at the receiver.

\section{Thesis Outline} \label{sec:Outline}
The rest of the thesis is organized as follows. In Chapter \ref{ch:DefinitionsAndNotation} we define rateless codes and provide related definitions and notation. In Chapter \ref{ch:PreviousResults} we survey previous results related to universal communication and rateless codes. In Chapter \ref{ch:KnownChannel} we treat the case of known channel, for which we obtain an achievable rate using rateless codes. We also prove a converse theorem showing that this rate is asymptotically optimal, and we analyze the rate of convergence. The case of unknown channel is examined in Chapter \ref{ch:UnknownChannel}, where we develop a universal decoder and analyze its performance for a general DMC. In Chapter \ref{ch:Extensions} we extend the coding scheme for the case of message sets with non-equiprobable messages, and we also show how rateless coding can be used for problems with side information. Chapter \ref{ch:Summary} concludes the thesis.

\chapter{Definitions and Notation}
\label{ch:DefinitionsAndNotation}
Throughout this thesis, random variables will be denoted by capital letters and their realizations by the corresponding lowercase letters. Vectors are denoted by superscript that indicate their length, for instance $X^n = [X_1,\ldots,X_n]$. Unless otherwise stated, all logarithms are taken to the base of 2. We focus on communication over a discrete memoryless channel (DMC) characterized by a transition probability $p(y|x)$, $x \in \mathcal{X} , y \in \mathcal{Y}$, where $\mathcal{X}$ and $\mathcal{Y}$ are the input and output alphabets of the channel, respectively. With a slight abuse of notation, we use $p(\cdot|\cdot)$ also to denote the joint transition probabilities of the channel, thus $p(y^n|x^n) = \prod_{i=1}^n p(y_i|x_i)$. The capacity of the channel (in bits per channel use) in conventionally defined as $C = \max_{q(x)}I(X;Y)$, where $I(X;Y)$ is the mutual information between the input of the channel and its output, and the maximization is over all channel input priors $q(x)$. If $|\mathcal{X}| = |\mathcal{Y}|$, and $p(y|x)=1$ if $x=y$ and $p(y|x)=0$ otherwise, then the channel is said to be noiseless, and in that case $C = \log |\mathcal{X}|$. We also assume that a noiseless feedback exists from the receiver to the transmitter.

A rateless code has the following elements:
\begin{enumerate}
\item Message set $\mathcal{W}$ containing $M$ messages. Without the loss of generality we assume that $\mathcal{W}=\{1,\ldots,M\}$, with corresponding probabilities $\pi(1),\ldots,\pi(M)$. Occasionally, we define $K=\log M$ as the number of bits conveyed in a message.
\item Codebook $\mathcal{C} = \{\mathbf{c}_i\}_{i=1}^M$, where each codeword $\mathbf{c}_i \in \mathcal{X}^\infty$ is generated by drawing i.i.d. symbols according to a prior $q(x), x\in\mathcal{X}$.
\item Set of encoding functions $f_n:\mathcal{W} \rightarrow \mathcal{X}$, $n \geq 1$.
\item Set of decoding function $g_n:\mathcal{Y}^n \rightarrow \mathcal{W} \cup \{0\}$, $n \geq 1$.
\end{enumerate}
Unlike conventional codes, for which the rate is a fundamental property, the above description does not specify a working rate--hence the term \emph{rateless code}. To encode a message $w \in \mathcal{W}$, the encoder starts transmitting the codeword $\mathbf{c}_w$ over the channel. Upon receiving each channel output, the decoder can either decide on one of the messages $\hat{w}$ or decide to wait for further channel outputs, returning `$0$'. Through feedback, the decoder's decision is known to the encoder, which correspondingly decides whether to transmit further symbols from $\mathbf{c}_w$ or to proceed to the next message. We note that two different forms of feedback can be assumed here: channel feedback and decision feedback. In channel feedback, the encoder at time instance $t$ observes $Y^{t-1}$, the channel outputs so far, and by imitating the decoder's operation it becomes aware of any decision made by the decoder. In decision feedback, the encoder is only informed that a decision has been made, and it can proceed to the next message. While channel feedback requires no intervention from the decoder in the feedback process, it essentially assumes that the feedback channel has the same bandwidth as the main channel. Decision feedback, in contrast, requires only one feedback bit per symbol.

We conclude this section with a few definitions required for the next sections.
\begin{definition} \label{def:StoppingTime}
A \emph{stopping time} $T$ of a rateless code is a random variable defined as
\begin{equation} \label{eq:StoppingTimeDef}
    T = \min\{n: g_n(Y^n) \neq 0\}
\end{equation}
\end{definition}
\begin{definition} \label{def:EffectiveRate}
An \emph{effective rate} $R$ of a rateless code is defined as
\begin{equation} \label{eq:EffectiveRateDef}
    R = \frac{\log M}{\E\{T\}}
\end{equation}
where $\E\{T\}=\E_q \{ \E_p \{T\}\}$, i.e. the averaging is done over all possible codebooks and channel realizations.
\end{definition}
Using the definition of stopping time, we can define the error event as the case in which the decoder stops, deciding on the wrong message. The error event conditioned on a particular message is defined as
\begin{equation} \label{eq:EmDef}
    E_w = \{\hat{W} \neq w \ | \ W = w \}
\end{equation}
where $\hat{W}=g_{_T}(Y^T)$.
The average error probability for the entire message set is therefore
\begin{equation} \label{eq:PeDef}
    P_e = \sum_{w=1}^M \pi(w) \cdot \Pr\{E_w\}
\end{equation}
\begin{definition}
For a given DMC, an $(R,M,\epsilon)$-code is a rateless code with effective rate $R$, containing $M$ messages and error probability $P_e \leq \epsilon$.
\end{definition}

\chapter{Previous Results}
\label{ch:PreviousResults}
As noted, the rateless coding scheme is a special case of communication over a channel with feedback. Shannon \cite{ShannonZEC} proved that the capacity of a DMC is not increased by adding feedback. However, adding feedback \emph{can} increase the zero-error capacity of the channel. In his well known paper, Burnashev \cite{Burnashev} investigated the effect of feedback in communication over a DMC by analyzing the error exponent of such channel. Introducing the notion of random transmission time, Burnashev obtained a bound on the mean transmission time for a fixed error probability, from which he derived the error exponent\footnote{Referred to as \emph{reliability function}.} for a DMC with feedback. He also proved a converse theorem showing that the expected transmission time, hence also the error exponent, are asymptotically optimal. (That is, they coincide with the results of the converse theorem as the size of the message set grows to infinity.) The main result of \cite{Burnashev} is the following theorem.
\begin{untheorem}[Burnashev \cite{Burnashev}]
The optimum error exponent for a DMC with noiseless feedback is
\begin{equation}\label{eq:BurnasheErrorExp}
    \lim_{M \to \infty} -\frac{1}{\E\{T\}} \log P_e = C_1\left(1 - \frac{R}{C} \right), \qquad 0 \leq R \leq C
\end{equation}
where $T$ is the transmission time, $R$ is defined in \eqref{eq:EffectiveRateDef} and
\begin{equation}\label{eq:C1Def}
    C_1 \triangleq \max_{(x,x') \in \mathcal{X} \times \mathcal{X}} D\left(p(\cdot|x)||p(\cdot|x')\right)
\end{equation}
\end{untheorem}

Examining \eqref{eq:BurnasheErrorExp} we can observe that whenever $R \geq C$, the error exponent vanishes, which concurs with Shannon's result \cite{ShannonZEC}. Moreover, whenever the channel has at least two inputs that are completely distinguishable from one another, i.e. $p(y|x)>0$ and $p(y|x')=0$ for some $x,x' \in \mathcal{X}$ and $y \in \mathcal{Y}$, it holds that $D\left(p(\cdot|x)||p(\cdot|x')\right) \to \infty$ and hence also $C_1 \to \infty$ for that channel. Therefore, the error exponent in that case is infinite at \emph{every} rate below the channel capacity, which implies that the zero-error capacity coincides with the channel capacity $C$.

Also for the case feedback channels, Shulman \cite{Nadav} developed a coding scheme providing reliable communication over an unknown channel, without compromising the rate. Introducing the concept of \emph{static broadcasting}, which is based on random codebook and universal sequential decoder, he demonstrated that it is possible to achieve vanishing error probability at rate that tends to the capacity of the channel as the size of the message set grows indefinitely. Furthermore, Shulman showed that even if the statistics of the information source is unknown to the transmitter, this scheme achieves the optimal decoding length that would have been achieved if the source were compressed by an optimal source-encoder and the channel were known at both ends. More formally, if $K$ information bits of a source $S$ were to be transmitted over an unknown channel $W$, then the average decoding length satisfies
\begin{equation}\label{eq:DecodingLengthNadav}
    \lim_{K \to \infty} \frac{\E\{T\}}{K} = \frac{H(S)}{I(P;W)}
\end{equation}
where $P$ is the codebook generation prior and $I(P;W)$ is the mutual information between the input and the output of the channel $W$ when the input is drawn according to distribution $P$.

Shulman also used the coding scheme for source-encoding of correlated sources. He demonstrated that using static broadcasting, it is possible to achieve the Slepian-Wolf optimal rate region. Combining all into one communication scheme, the achievable decoding length is
\begin{equation}\label{eq:DecodingLengthNadav}
    \lim_{K \to \infty} \frac{\E\{T\}}{K} = \frac{H(S|Z)}{I(P;W)}
\end{equation}
where $Z$ is the side information at the decoder. Shulman's work has been the main inspiration for this research.

For the case of unknown channel, Tchamkerten and Telatar in \cite{Telatar} used a rateless coding scheme similar to the one defined in Chapter \ref{ch:DefinitionsAndNotation}, where the stopping condition is that the mutual information between (at least) one of the codewords and the channel output sequence exceeds a certain time-dependent threshold. The authors proved that this scheme can achieve the capacity of a general DMC.\footnote{Since no assumption has been made on the capacity-achieving prior, authors only demonstrated that the rate approaches $I(PQ)$, where $P$ is the codebook generation prior and $Q$ is the transition probability of the channel.} Moreover, they demonstrated that for the class of binary symmetric channel with crossover probabilities $L \in [0,1/2)$, this coding scheme can achieve Burnashev's exponent at a rate bounded by any fraction of the channel capacity. The latter result is obtained by using a second coding phase, in which the transmitter indicates whether the decoder's decision is correct (an \emph{Ack/Nack} phase). Tchamkerten and Telatar also demonstrated that for the class of $Z$ channels with parameter $L \in [0,1)$, the achievable rate can be arbitrarily close to the channel capacity, while the error exponent is infinite. The latter result also coincides with Burnashev's exponent ($C_1$ in \eqref{eq:BurnasheErrorExp} is infinite in this case), since error-free communication is attainable for the $Z$ channel.

We note that all the above-mentioned results were asymptotic in the size of the message set. Recently, Polyanskiy, Poor and Verd{\'u} in \cite{PPV} introduced non-asymptotic results for communication over DMC with feedback. Through the use of variable-rate coding and sequential decoding they obtained upper and lower bounds for the maximal message set size for fixed bounds on the error probability and mean decoding length. The authors showed that for an error probability constraint $P_e \leq \epsilon$ and mean decoding length constraint $\E \{T\} \leq \ell$, the maximal message set size $M^*(\ell,\epsilon)$ satisfies
\begin{equation}\label{eq:PolyanskiyUpperLower}
    \frac{\ell C}{1-\epsilon} - \log \ell + O(1) \leq \log M^*(\ell,\epsilon) \leq \frac{\ell C}{1-\epsilon} + O(1)
\end{equation}
The setting of \cite{PPV}, as well as the coding scheme, is similar to the one defined later in Chapter \ref{ch:KnownChannel}. However, while in \cite{PPV} the optimization is on $M$, for fixed $\epsilon$ and $\ell$, we fix $\epsilon$ and $M$ and find the optimum mean decoding length. The analysis is slightly different, but the results of Chapter \ref{ch:KnownChannel} comply with \cite{PPV}. The analysis in Chapter \ref{ch:KnownChannel}, coming next, lays the ground for the derivation of our novel results for the case of unknown channel.

\chapter{Rateless Coding -- Known Channel}
\label{ch:KnownChannel}
\section{Sequential Decoder} \label{sec:SequentialDecoder}
We begin by introducing a rateless coding scheme for noisy channels and analyzing its effective rate, under certain constraints on the size of the message set and the error probability. As will be shown in the sequel, the effective rate is closely related to the channel capacity $C$. More precisely, we will show that under the conventional setting, in which the size message set is taken to infinity, the effective rate coincides with the capacity of the channel.

Consider a discrete memoryless source with a set of $M$ equiprobable messages, i.e. $\pi(i)=1/M$, $i=1,\ldots,M$. We use a rateless code as defined in Section \ref{ch:DefinitionsAndNotation}, where each codeword $\mathbf{c}_i$, $i=1,\ldots,M$ is generated by drawing i.i.d. symbols according to $q(x)$, the capacity-achieving prior of the channel. The source of randomness generating the codewords is shared by the encoder and the decoder, so that the codebook in known at both ends. The decoder uses the following decision rule:
\begin{equation} \label{eq:ChannelDecoderLin}
    g_n(y^n)= \begin{cases}
                w, \ \prod_{k=1}^n p(c_{w,k}|y_k) \geq A \cdot \prod_{k=1}^n q(c_{w,k}) \\
                0, \ \text{if no such $w$ exists}
              \end{cases}
\end{equation}
where $\{c_{w,k}\}_{k=1}^{\infty}$ are the symbols in $\mathbf{c}_w$. If the threshold crossing condition in \eqref{eq:ChannelDecoderLin} is satisfied by more than one codeword, we randomly choose one of them and declare an error. We note here that similar decoders have been proposed by Polyanskiy \cite{PPV} and Burnashev \cite[Ch.3]{Burnashev}.
The decision rule at \eqref{eq:ChannelDecoderLin} can be equivalently written as
\begin{equation} \label{eq:ChannelDecoderLog}
    g_n(y^n)= \begin{cases}
                w, \ z_{w,1}+\ldots+z_{w,n} \geq a \\
                0, \ \text{if no such $w$ exists}
              \end{cases}
\end{equation}
where
\begin{equation}
    z_{w,k} = \log \frac{p(c_{w,k}|y_k)}{q(c_{w,k})}, \qquad k=1,\ldots,n
\end{equation}
and we define $a = \log A$.

The above-described coding scheme can be summarized as follows. Having selected a message, the encoder starts transmitting an infinite-length random codeword corresponding to that message. The decoder sequentially receives symbols from this codeword that passed through the channel, and at each time instant $k$ calculates $z_{w,k}$ for $w={1,\ldots,M}$. It then updates a set of $M$ accumulators, each corresponding to a possible message, and checks whether any of those crossed a prescribed threshold $a$. If neither of the counters crossed the threshold, `$0$' is returned and the decoder waits for the next channel output; if exactly one counter crossed the threshold, the decoder makes a decision; and if more that one threshold crossing occurred, an error is declared. In the two latter cases, the encoder proceeds to the next codeword.

For the above-described scheme we have the following theorem.
\begin{theorem} \label{Theorem1}
For the decoder in \eqref{eq:ChannelDecoderLog} with $P_e \leq \epsilon$, the following effective rate is achievable:
\begin{equation}
    R = \frac{C}{1+ \frac{C - \log \epsilon}{\log M}} \label{eq:AchievableRateChannelDec}
\end{equation}
\end{theorem}

\begin{proof}
Since $T$ is a stopping time of the i.i.d. sequence $Z_1,Z_2,\ldots$, Wald's equation \cite{Wald} implies
\begin{equation} \label{eq:ETForWald}
    \E\{T\} = \frac{\E \{Z_1 + \ldots + Z_T\}}{\E\{Z\}},
\end{equation}
where $\E\{Z\}$ is the expectation of a single sample $Z_i$. If $X_i$ and $Y_i$ are input and output of the channel, respectively, then by the definition of $Z_i$ we have
\begin{equation} \label{eq:EZForWald}
    \E\{Z\} = \E\{Z_i\} = \E \left\{\frac{p(X_i|Y_i)}{q(X_i)}\right\} = C.
\end{equation}
Furthermore, since the stopping condition was not fulfilled time instant $T-1$ we have
\begin{equation}
    Z_1 + \ldots + Z_{T-1} < a
\end{equation}
which implies
\begin{equation} \label{eq:SumZForWald}
    Z_1 + \ldots + Z_T < a + Z_T
\end{equation}
Combining \eqref{eq:ETForWald}, \eqref{eq:EZForWald} and \eqref{eq:SumZForWald} we obtain
\begin{equation} \label{eq:WaldC}
    \E\{T\} < \frac{a+C}{C}.
\end{equation}

We now tune the threshold parameter $a$ to meet the error probability requirement. Suppose that the stopping time of the correct codeword is $T_w$. An error occurs if a competing codeword $\mathbf{c}_{w'}$, independent of $\{Y_k\}_{k=1}^{\infty}$, crosses the threshold before $\mathbf{c}_w$ does. Thus,
\begin{align}
    \Pr \{E_w\} &= \Pr \left\{ \bigcup_{w' \neq w} \bigcup_{t=1}^{T_w}
    \left\{ \frac{\prod_{k=1}^{t} p(C_{w',k}|Y_k)}{\prod_{k=1}^{t} q(C_{w',k})} > A \right\} \right\} \\
                &\leq (M-1) \Pr \left\{ \bigcup_{t=1}^{T_w}
    \left\{ \frac{\prod_{k=1}^{t} p(X_k|Y_k)}{\prod_{k=1}^{t} q(X_k)} > A \right\} \right\} \label{eq:ErrorProbUB} \\
                &\leq (M-1) \Pr \left\{ \bigcup_{t=1}^{\infty}
    \left\{ \frac{\prod_{k=1}^{t} p(X_k|Y_k)}{\prod_{k=1}^{t} q(X_k)} > A \right\} \right\} \label{eq:ErrorProbInf}
\end{align}
where \eqref{eq:ErrorProbUB} follows from the union bound for an arbitrary series $\{X_k\}_{k=1}^{\infty}$ drawn i.i.d. from $q(x)$, independently of $\{Y_k\}_{k=1}^{\infty}$. Note that the bound in \eqref{eq:ErrorProbInf} represents the probability that a randomly-chosen codeword will exceed the threshold at any time instant. Define a sequence of random variables
\begin{equation}\label{eq:uidef}
    U_t = \begin{cases}
                \frac{p(X_t|Y_t)}{q(X_t)}, \ \prod_{k=1}^{t-1} U_k \leq A \\
                1, \ \text{otherwise}
              \end{cases}
\end{equation}
If at instant $t$ the threshold at \eqref{eq:ErrorProbInf} is exceeded for the first time, then we have $U_k=p(X_k|Y_k)/q(X_k)$ for $k=1,\ldots,t$ and $U_k=1$ for all $k>t$. Therefore, it is easy to see that
\begin{equation}
    \bigcup_{t=1}^{\infty}
    \left\{ \frac{\prod_{k=1}^{t} p(X_k|Y_k)}{\prod_{k=1}^{t} q(X_k)} > A \right\}
    \Leftrightarrow \prod_{t=1}^{\infty} U_t > A
\end{equation}
We can also see that $\E\{U_t\}=1$ for \emph{all} $t$ because
\begin{equation*}
    \E \left\{ U_t|\prod_{k=1}^{t-1} U_k > A \right\} = 1
\end{equation*}
since $U_t = 1$ deterministically in this case, and
\begin{align}
    \E \left\{ U_t|\prod_{k=1}^{t-1} U_k \leq A \right\} &= \E \left\{ \frac{p(X_t|Y_t)}{q(X_t)} \right\} \\
    &= \E \left\{ \E \left\{ \frac{p(X_t|Y_t)}{q(X_t)}|Y_t \right\} \right\} \\
    &= \E \left\{ \sum_{x \in \mathcal{X}} \frac{p(x|Y_t)}{q(x)} \cdot q(x) \right\} \label{eq:IndependentX} \\
    &= 1
\end{align}
where \eqref{eq:IndependentX} follows since $X_t$ and $Y_t$ are independent. For an arbitrary $N$ we have
\begin{align}\label{eq:ProdU}
    \E \left\{\prod_{t=1}^N U_t \right\} &= \E \left\{ \E \left\{\prod_{t=1}^N U_t | \prod_{t=1}^{N-1} U_t \right\} \right\} \\
    &= \E \{ U_N \} \cdot \E \left\{\prod_{t=1}^{N-1} U_t \right\} \\
    &= \E \left\{\prod_{t=1}^{N-1} U_t \right\} = \ldots = \E\{U_1\} = 1
\end{align}
Since the above holds for all $N$, we also have
\begin{equation}\label{eq:InfProdU}
    \E \left\{\prod_{t=1}^\infty U_t \right\} = 1
\end{equation}
Returning to \eqref{eq:ErrorProbInf}, we get
\begin{align}
   \Pr \{E_w\} &\leq  (M-1) \Pr \left\{ \bigcup_{t=1}^{\infty}
    \left\{ \frac{\prod_{k=1}^{t} p(X_k|Y_k)}{\prod_{k=1}^{t} q(X_k)} > A \right\} \right\} \label{eq:BoundErrorProbKnownCh} \\
    &= (M-1) \Pr \left\{ \prod_{t=1}^{\infty} U_t > A \right\} \\
    &\leq \frac{M-1}{A} \label{eq:Markov}
\end{align}
where \eqref{eq:Markov} follows from \eqref{eq:InfProdU} and Markov's Inequality.
Since the above holds for all $w \in \mathcal{W}$, we also have
\begin{equation}
    P_e \leq \frac{M-1}{A}
\end{equation}

By choosing $a = \log M - \log \epsilon $, or equivalently $A = M/\epsilon$, we secure that $P_e < \epsilon$. Substituting $a$ into \eqref{eq:WaldC} and using Definition \ref{def:EffectiveRate}, we obtain \eqref{eq:AchievableRateChannelDec}.
\end{proof}

It is important to note that the encoding operation is independent of the working rate; the encoder needs to know the channel law only to generate the codebook. However, if the channel is known to belong to a family for which the capacity-achieving prior is known (e.g. the uniform prior for symmetric channel), then the optimal rate can be achieved even when the encoder is uninformed on the channel law. Furthermore, from a practical point of view, using the uniform prior instead of the capacity-achieving prior is known to perform relatively well in many cases. For instance, using a uniform prior in a binary channel will lose at most 6\% of the capacity (see \cite[Ch.5]{Nadav}).

\section{Coding Theorem for Known Channel} \label{sec:CodingThmKnownChannel}
We will now use the coding scheme developed in Section \ref{sec:SequentialDecoder} to prove the main result for rateless channel coding. For a fixed error probability, we will obtain an achievable rate using rateless codes. Then, we will prove that this rate is within $O(\log \log M / \log M)$ from the optimal rate achievable with this error probability. Before we get to the main theorem, we prove the following lemma, which facilitates some refinement in the achievable rate.

\begin{lemma} \label{Lemma1}
Suppose that an $(R,M,\epsilon)$-code exists for a DMC. Then for any $0<\alpha<1$, there also exists an $(R',M,\epsilon')$-code for the same channel, where
\begin{eqnarray}
    R' & = & (1-\alpha)^{-1}R \\
    \epsilon' & = & \alpha + \epsilon - \alpha \epsilon
\end{eqnarray}

\end{lemma}

\begin{proof}
To show that the triplet $(R',M,\epsilon')$ is achievable, we use the original code with randomized decision-making at the decoder. For each transmitted message, the decoder either terminates the transmission immediately and declares an error, with probability $\alpha$, or uses the original decision rule. Denote the stopping time of the original decoder and the modified one by $T$ and $T'$, respectively. The expected decision time of the modified decoder is
\begin{equation}
    \E\{T'\} = (1-\alpha) \E\{T\},
\end{equation}
which implies
\begin{equation}
    R' = (1-\alpha)^{-1} R.
\end{equation}
The error event in the modified scheme is a union of two non-mutually-exclusive events: error in the original decoder and the event of early termination. The probability of this union is
\begin{equation}
    \epsilon' = \alpha + \epsilon - \alpha \epsilon.
\end{equation}
Finally, we note that the number of messages in the codebook remains unchanged---which completes the proof of the lemma.
\end{proof}

\begin{theorem}
For rateless codes, the following rate is achievable:
\begin{equation} \label{eq:AchievableRate}
   R' = \begin{cases}
                \frac{1 - 1/\log M}{1+ \frac{C + \log \log M}{\log M}} \cdot \frac{C}{1-\epsilon} & \epsilon > 1/\log M \\
                \frac{C}{1+ \frac{C - \log \epsilon}{\log M}} & \epsilon \leq 1/\log M \\
              \end{cases}
\end{equation}
\end{theorem}

We note that if $\epsilon$ is fixed and $M$ is large enough so that $\epsilon > 1/\log M$, the achievable rate has the following asymptotics:
\begin{equation} \label{eq:AchievableRateAsym}
    R' = \frac{C}{1-\epsilon} \cdot \left(1 - O \left( \frac{\log \log M}{\log M} \right) \right)
\end{equation}
\begin{proof}
Theorem \ref{Theorem1} implies that the triplet $(R,M,\delta)$ is achievable for all $0 < \delta < 1$, where
\begin{equation}
    R = \frac{C}{1+ \frac{C - \log \delta}{\log M}}
\end{equation}
By Lemma \ref{Lemma1}, we can also achieve $(R',M,\delta')$, where
\begin{eqnarray}
    R' & = & \frac{C}{(1-\alpha)\left(1+ \frac{C - \log \delta}{\log M}\right)} \\
    \delta' & = & \alpha + \delta - \alpha \delta
\end{eqnarray}
for all $0 < \alpha < 1$. By choosing
\begin{equation}
    \alpha = \frac{\epsilon - \delta}{1 - \delta}
\end{equation}
we obtain
\begin{eqnarray}
    R' & = & \frac{1 - \delta}{1+ \frac{C - \log \delta}{\log M}} \cdot \frac{C}{1-\epsilon} \\
    \delta' & = & \epsilon
\end{eqnarray}
Since the foregoing analysis holds for all $0 < \delta < \epsilon$, we can choose $\delta = \min\{\epsilon, 1/\log M\}$ to obtain \eqref{eq:AchievableRate}.
\end{proof}

\begin{remark}
If $\epsilon \leq 1/ \log M$, the choice $\delta = \epsilon$ implies $\alpha = 0$, that is, no randomization at the decoder. This result could be anticipated, since the randomized decoder trades rate for reliability: it obtains a better effective rate with some compromise on the error probability. Hence, whenever the error probability constraint is more important than the working rate -- randomization can only worsen matters.
\end{remark}

\section{Error Exponent}
Theorem 2 in the previous section provides a relation between the working rate and the allowed error probability. We will now investigate this dependency in the regime of low error probability by developing the error exponent induced by this coding scheme. Assuming that a low error probability is required, randomization at the decoder is inapplicable, so \eqref{eq:AchievableRate} can be rewritten as
\begin{equation}
    - \frac{R}{\log M} \log \epsilon = C - R - \frac{CR}{\log M}
\end{equation}

Recall that $R = \log M / \E\{T\}$, so
\begin{equation}
    - \frac{\log \epsilon}{\E\{T\}} = C - R - \frac{CR}{\log M} \triangleq E(R)
\end{equation}

We can see that the error exponent is a \emph{linear} function of the rate, which is also the case in Burnashev's analysis \eqref{eq:BurnasheErrorExp} (albeit with a different coefficient). Furthermore, as $M$ grows, the error exponent converges to $C-R$ and the convergence is dominated by a term of order $O(1/\log M)$, or $O(1/K)$. This term can be interpreted as a penalty for using a finite message set.

\section{Weak Converse}
In the previous section we have seen that if we use a codebook with $M$ messages and allow an error probability $P_e \leq \epsilon$, then we can achieve an effective rate with the following asymptotics:
\begin{equation}
    R' = \frac{C}{1-\epsilon} \cdot \left(1 - O \left( \frac{\log \log M}{\log M} \right) \right)
\end{equation}

We will now prove that under the above constraints on the message set and the error probability, the best achievable rate has the same asymptotics. In other words, the achievable rate at \eqref{eq:AchievableRate} converges to the optimal rate, and the convergence is dominated by a term of order $O(1/\log M)$.

\begin{theorem}
Given a decoder with random \footnote{Fixed stopping time is a private case of random stopping time, in which $T$ takes only one value.} stopping time $T$, any rate for which the probability of error does not exceed $\epsilon$ satisfies
\begin{equation}
    R' \leq \frac{C}{1-\epsilon} \cdot \left(1 + O \left( \frac{1}{\log M} \right) \right).
\end{equation}
\end{theorem}

\begin{proof}
Define
\begin{equation}
    \mu(n) = H(W|Y^n) + nC.
\end{equation}
By \cite[Lemma 2]{Burnashev}) we have
\begin{equation}
    \E \{\mu(n+1)|Y^n\} - \mu(n) = \E \{H(W|Y^{n+1})-H(W|Y^n)|Y^n\} + C \geq 0
\end{equation}
which implies that $\mu(n)$ is a submartingale with respect to the process $\{Y_k\}_{k=1}^{\infty}$. Therefore we have
\begin{align}
    \log M &= H(W) = \mu(0) \nonumber \\
           &\leq \E \{\mu(T)\} \nonumber \\
           &= \E \{H(W|Y^T)\} + C \cdot \E\{T\} \label{eq:SubmartingalIneq}
\end{align}
Furthermore, by \cite[Lemma 1]{Burnashev} we have
\begin{align}
    \E\{H(W|Y^T)\} &\leq h \left(P_e\right) + P_e \cdot \log (M-1) \nonumber \\
             &< 1 + \epsilon \cdot \log M \label{eq:Fano}
\end{align}
where \eqref{eq:Fano} follows from the requirement $P_e \leq \epsilon$, and from an upper bound on the binary entropy function. Combining \eqref{eq:SubmartingalIneq} and \eqref{eq:Fano} we obtain
\begin{equation} \label{eq:BoundLogM}
    \log M < 1 + \epsilon \cdot \log M + C \cdot \E\{T\}
\end{equation}
which implies
\begin{equation} \label{eq:RPrimeWithET}
    R' = \frac{\log M}{\E\{T\}} < \frac{C}{1-\epsilon} \cdot \left(1 + \frac{1}{C \cdot \E\{T\}}\right)
\end{equation}
Furthermore, from \eqref{eq:BoundLogM} we can see that
\begin{equation}
    C \cdot \E\{T\} > (1 - \epsilon) \cdot \log M - 1
\end{equation}
and therefore \eqref{eq:RPrimeWithET} can be replaced by
\begin{equation}
    R' = \frac{\log M}{\E\{T\}} \leq \frac{C}{1-\epsilon} \cdot \left(1 + O\left(\frac{1}{\log M}\right)\right) \label{eq:UpperBound}
\end{equation}
\end{proof}

\begin{remark}
While \eqref{eq:AchievableRate} approaches \eqref{eq:UpperBound} for large $M$, the upper bound is not tight for a finite $M$. Note that the converse used here is ``weak'', in that it is based on Fano's inequality, which is known to be loose in many cases. We conjecture that a strong converse can be found, which will be tighter (i.e. closer to \eqref{eq:AchievableRate}) even in the non-asymptotic realm.
\end{remark}

\begin{remark}
Equation \eqref{eq:AchievableRateAsym}, the achievable rate, is essentially equivalent to the left-hand side of \cite[Eq.18]{PPV}, and equation \eqref{eq:UpperBound}, the upper bound on the rate, is equivalent to the right-hand side of that equation. Note, however, that the formulation is slightly different: in \cite{PPV} size of the message set $M$ is optimized with constraint on the maximal transmission time, while here $M$ is fixed and the transmission time is minimized.
\end{remark}

\section{Further Discussions}
\subsection{Application for Gaussian Channels}
While the analysis in Sections \ref{sec:SequentialDecoder} and \ref{sec:CodingThmKnownChannel} is done for discrete channels, it can be easily extended to memoryless Gaussian channels. Suppose that $X_t$ and $Y_t$ are the input and output of an additive white Gaussian noise channel at time instant $t$, i.e.
\begin{equation} \label{eq:AWGNDef}
    Y_t = X_t + V_t, \qquad t=1,2,\ldots
\end{equation}
where $\{V_t\}_{t=1}^\infty$ is a sequence of i.i.d. Gaussian RV's with zero mean and a known variance. The encoding and the decoding processes, as well as the expression for the resulting effective rate, are similar to those of the DMC, where $q(\cdot)$ is the codebook generation PDF and $p(\cdot|\cdot)$ is the transition PDF of the backward channel.

Specifically, consider the above-described setting where $V_k \sim N(0,\theta)$. Suppose that the codebook is Gaussian with power constraint $P$, i.e. $C_{m,k} \sim N(0,P)$ for all $m,k$. (Here again, $C_{m,k}$ is the $k$-th symbol of the $m$-th codeword.) The decoding rule is given by \eqref{eq:ChannelDecoderLin}, where
\begin{align}
    p(x|y) &= \left( 2\pi \condvar \right)^{-1/2} \exp \left\{ -\frac{1}{2 \cdot \condvar} \left( x - \Wiener \cdot y \right)^2 \right\} \\
    q(x) &= \left( 2\pi P \right)^{-1/2} \exp \left\{ -\frac{x^2}{2P} \right\}
\end{align}
The effective rate of the decoder is given in \eqref{eq:AchievableRateChannelDec}, where
\begin{equation}
    C = \frac{1}{2} \log \left( 1 + \frac{P}{\theta}\right)
\end{equation}

\subsection{Limited Feedback Channel}
In the forgoing analysis, we assumed that the feedback channel must be used once per each main channel use. In practice, however, it may be desirable to reduce the amount of data transmitted over the feedback channel. For instance, in the case of broadcasting to multiple users, the upstream channel may have a more stringent bandwidth constraint as it must be accessed by all users. It is therefore interesting to see how lowering the frequency of the feedback affects the performance of the rateless coding scheme. Suppose that we want to use the feedback channel only once per $s$ received symbols. The maximal number of excess symbols transmitted over the main channel (i.e. the number of symbols transmitted after a decoder without feedback limitation would acknowledge the message) is $s-1$, which implies an effective rate of
\begin{equation}\label{eq:RateLimitedFB}
    R = \frac{C}{1+ \frac{(s-1)C - \log \epsilon}{\log M}}
\end{equation}
From \eqref{eq:RateLimitedFB} we see that limiting the feedback frequency has negligible effect if either $s \ll (-\log \epsilon)/C$ or $s \ll (\log M)/C$. In the former case, the required confidence level is high, and in the latter case the messages are long. That is, in both cases the codewords are long with respect to the capacity of the channel, which implies long transmission time. Therefore, in both cases the excess decoding time is small compared to the entire transmission length, and the effect of the limiting the feedback is negligible.

\chapter{Rateless Coding -- Unknown Channel}
\label{ch:UnknownChannel}
In Chapter \ref{ch:KnownChannel} we assumed that the communication channel, characterized by $p(y|x)$, is known at the receiver end. Assume now, that the underlying channel is unknown to the receiver. The capacity of the channel is known to be achievable in this scenario using sequential versions of the Maximal Mutual Information (MMI) decoder \cite{Nadav}, \cite{Telatar}. However, while these schemes provide reliable communication at rate equal to the channel capacity, they assume that the size of the message set $M$ is infinite. In this chapter we try to answer the question whether universal communication is feasible with a finite message set, and if it is, what rates are achievable? As we shall see shortly, it is possible to achieve reliable communication over an unknown channel even when the message set is finite, and we can also bound the rate degradation due to lack of information about the channel law.

\section{Achievable Rate for an Unknown Channel} \label{sec:RateUnknownChannel}
Suppose that we wish to communicate over a DMC with unknown (backward) transition probabilities
\begin{equation}
    \theta_{ij} = \Pr\{X=i|Y=j\}, \qquad i = 1,\ldots,|\mathcal{X}| \qquad j = 1,\ldots,|\mathcal{Y}|
\end{equation} \label{eq:ThetaParamsDef}
We use a coding scheme similar to the one described in Chapter \ref{ch:KnownChannel} with the following modification. Instead of using the true transition probability $p_{\tv}(x^t|y^t)$, which is unknown to the decoder, we use a \emph{universal} probability assignment defined as
\begin{equation}\label{eq:UniversalProb}
    p_U (x^t|y^t) \triangleq \int_{\Lambda} w(\tv') p_{\tv'} (x^t|y^t) d\tv'
\end{equation}
where
\begin{equation}
    \Lambda = \left\{\tv' \in [0,1]^{\XY} \ | \ \sum_{i=1}^{|\mathcal{X}|} \theta'_{ij} = 1, \quad j = 1,\ldots,|\mathcal{Y}|\right\}
\end{equation}
and the weight function $w(\cdot)$ is chosen to be Jeffreys Prior \footnote{This is also a special case of Dirichlet Distribution.}, i.e.
\begin{equation} \label{eq:JeffreysPrior}
    w(\tv') = \frac{1}{\BXY \sqrt{\prod_{i,j} \theta'_{ij}}}
\end{equation}
where
\begin{equation} \label{eq:DefBXY}
    \BXY = \int_{\Lambda} \frac{d\tv'}{\sqrt{\prod_{i,j} \theta_{ij}}}
\end{equation}

\begin{remark}
While the unknown channel is usually characterized by a set of transition probabilities
\begin{equation*}
    \tilde{\theta}_{ij} = \Pr\{Y=j|X=i\}, \qquad i = 1,\ldots,|\mathcal{X}| \qquad j = 1,\ldots,|\mathcal{Y}|
\end{equation*}
the entire derivation here is done for the \emph{backward} channel parameterization given in \eqref{eq:ThetaParamsDef}. However, this does not need to bother us since the entire analysis assumes a known input prior $q(x)$, and therefore given $\{\tilde{\theta}_{ij}\}$, the parameters in \eqref{eq:ThetaParamsDef} are well-defined. Moreover, the region $\tilde{\Lambda}$, induced by $\{\tilde{\theta}_{ij}\}$ and $q(x)$, is clearly contained in the region $\Lambda$. Therefore, if a coding scheme is universal with respect to all possible realizations of the backward channel, it is also universal w.r.t. all possible realizations of the forward channel.
\end{remark}

The universal probability assignment implies the following decoding rule, which is the universal counterpart of \eqref{eq:ChannelDecoderLin}:
\begin{equation} \label{eq:ChannelDecoderUnv}
    g_n(y^n)= \begin{cases}
                w, \ p_U(\mathbf{c}_w|y^n) \geq A \cdot q(\mathbf{c}_w) \\
                0, \ \text{if no such $w$ exists}
              \end{cases}
\end{equation}

In Chapter \ref{ch:KnownChannel} we used Wald's Identity to bound the expected transmission time, thereby obtaining an effective rate for the sequential decoder. Unfortunately, in the universal case $p_U(\cdot|\cdot)$ is not necessarily multiplicative, so $\log p_U(\cdot|\cdot)$ cannot be expressed as the sum of i.i.d. random variables. Therefore, the expected transmission time in the universal case cannot be calculated directly by applying Wald's identity. Nevertheless, as we shall see shortly, we can use the results for the known channel case to obtain an upper bound for the transmission time in the universal case.

The following lemma shows that given two sequences $x^t$ and $y^t$, the universal metric cannot be too far from the conditional probability assignment that is optimally fitted to $x^t$ and $y^t$.
\begin{lemma} \label{thm:UnvProbLemma}
For any two series $x^t$ and $y^t$ we have
\begin{equation}
    \log \frac{p_{\htv} (x^t|y^t)}{p_U (x^t|y^t)} \leq \frac{\left(\XX-1\right)\YY}{2} \log \frac{t}{2\pi} + \YY \Lkappa_{\XX} + \left( \frac{\XX^2 \YY}{4} + \frac{\XY}{2} \right) \log e
\end{equation}
where
\begin{equation} \label{eq:ThetaOpt}
    \htv = \arg \max_{\tv' \in \Lambda} p_{\tv'} (x^t|y^t)
\end{equation}
and we define
\begin{equation}
    \Lkappa_{\XX} = \log \frac{\Gamma(1/2)^{\XX}}{\Gamma(\XX/2)}
\end{equation}
\end{lemma}

\begin{proof}
Note that
\begin{align}
    p_{\htv}(x^t|y^t)&= \max_{\{\theta_{i,j}\}} \prod_{i,j} \theta_{i,j}^{N(x^t,y^t;i,j)} \\
             &= \max_{\{\theta_{i,1}\}} \prod_i \theta_{i,1}^{N(x^t,y^t;i,1)} \cdot \max_{\{\theta_{i,2}\}} \prod_i \theta_{i,2}^{N(x^t,y^t;i,2)} \cdot \ldots \cdot
             \max_{\{\theta_{i,\YY}\}} \prod_i \theta_{i,\YY}^{N(x^t,y^t;i,\YY)}
\end{align}
where
\begin{equation} \label{eq:NxyDef}
    N(x^t,y^t;i,j) = \left| \left\{k \ : \ (x_k,y_k)=(i,j) \right\} \right|
\end{equation}
Since both $w(\cdot)$ and $p_{\tv}$ are multiplicative functions, we also have
\begin{align}
    p_U (x^t|y^t) &= \int_{\Lambda} w(\tv') p_{\tv'} (x^t|y^t) d\tv' \\
                  &= \int_{\tilde{\Lambda}} w(\tilde{\tv}) \prod_i \tilde{\theta}_i^{N(x^t,y^t;i,1)} d\tilde{\tv} \cdot \int_{\tilde{\Lambda}} w(\tilde{\tv}) \prod_i \tilde{\theta}_i^{N(x^t,y^t;i,2)} d\tilde{\tv} \cdot \ldots \\
                  & \cdot \int_{\tilde{\Lambda}} w(\tilde{\tv}) \prod_i \tilde{\theta}_i^{N(x^t,y^t;i,\YY)} d\tilde{\tv}
\end{align}
where
\begin{equation}
    \tilde{\Lambda} = \left\{\tilde{\tv} \in [0,1]^{\XX} \ | \ \sum_{i=1}^{|\mathcal{X}|} \tilde{\theta}_i = 1, \right\}
\end{equation}

From \cite[Lemma 1]{Barron} we know that
\begin{equation}\label{eq:BarronsIneq}
    \log \frac{\max_{\{\theta_{i,j}\}} \prod_i \theta_{i,j}^{N(x^t,y^t;i,j)}}{\int_{\tilde{\Lambda}} w(\tilde{\tv}) \prod_i \tilde{\theta}_{i,j}^{N(x^t,y^t;i,j)} d\tilde{\tv}} \leq \frac{\XX-1}{2} \log \frac{t}{2\pi} + \Lkappa_{\XX} + \left( \frac{\XX^2}{4} + \frac{\XX}{2} \right) \log e
\end{equation}
for all $j = 1,\ldots,\YY$. Thus, we obtain
\begin{align}
    \log \frac{p_{\htv} (x^t|y^t)}{p_U (x^t|y^t)} &= \log \prod_j \frac{\max_{\{\theta_{i,j}\}} \prod_i \theta_{i,j}^{N(x^t,y^t;i,j)}}{\int_{\tilde{\Lambda}} w(\tilde{\tv}) \prod_i \tilde{\theta}_{i,j}^{N(x^t,y^t;i,j)} d\tilde{\tv}} \\
    &\leq \frac{\left(\XX-1\right)\YY}{2} \log \frac{t}{2\pi} + \YY \Lkappa_{\XX} + \left( \frac{\XX^2 \YY}{4} + \frac{\XY}{2} \right) \log e \\
    &=\frac{\left(\XX-1\right)\YY}{2} \log t + \beta
\end{align}
where we define
\begin{equation}\label{eq:BetaDef}
    \beta \triangleq \YY \Lkappa_{\XX} + \left( \frac{\XX^2 \YY}{4} + \frac{\XY}{2} \right) \log e - \frac{\left(\XX-1\right)\YY}{2} \log (2\pi)
\end{equation}
\end{proof}

We are now ready to prove the main theorem for rateless coding over an unknown channel.

\begin{theorem} \label{thm:UnvRate}
For the decoder in \ref{eq:ChannelDecoderUnv} with $P_e \leq \epsilon$, the following effective rate is achievable:
\begin{equation}\label{eq:UnvEffectiveRate}
    R = \frac{C \left( 1 - \frac{\hXY}{\log M \ln 2} \right)}{1 + \frac{C + \beta - \log \epsilon + \frac{\XY}{2} \left( \log \log M - \log C - \frac{1}{\ln 2} \right)}{\log M}}
\end{equation}
\end{theorem}

\begin{proof}
The stopping time in the above-described scheme is
\begin{equation} \label{eq:StoppingTimeUnv}
    T = \min \left\{ t: \frac{p_U (x^t|y^t)}{\prod_{k=1}^t q(x_k)} > A \right\}
\end{equation}
since
\begin{equation}
    \log p_U (x^t|y^t) = \log p_{\tv} (x^t|y^t) - \log \frac{p_{\tv} (x^t|y^t)}{p_U (x^t|y^t)}
\end{equation}
we have
\begin{align}
    T &= \min \left\{ t: \frac{p_U (x^t|y^t)}{\prod_{k=1}^t q(x_k)} > A \right\} \\
      &= \min \left\{ t: \log \frac{ \prod_{k=1}^t p_{\tv} (x_k|y_k)}{\prod_{k=1}^t q(x_k)} > \log A + \log \frac{p_{\tv} (x^t|y^t)}{p_U (x^t|y^t)} \right\} \\
      &\leq \min \left\{ t: \log \frac{ \prod_{k=1}^t p_{\tv} (x_k|y_k)}{\prod_{k=1}^t q(x_k)} > \log A + \log \frac{p_{\htv} (x^t|y^t)}{p_U (x^t|y^t)} \right\} \label{eq:ThetaOptUse} \\
      &< \min \left\{ t: \sum_{k=1}^t \log \frac{p_{\tv}(x_k|y_k)}{q(x_k)} > \log A + \frac{\XY}{2} \log t + \beta \right\} \label{eq:UnvProbLemmaUse}
\end{align}
where \eqref{eq:ThetaOptUse} follows since $p_{\htv} (x_k|y_k) \geq p_{\tv} (x_k|y_k)$ by definition \eqref{eq:ThetaOpt} and \eqref{eq:UnvProbLemmaUse} follows from Lemma \ref{thm:UnvProbLemma}.

From the same considerations as in the proof of Theorem \ref{Theorem1}, at the stopping time $T$ we necessarily have
\begin{equation} \label{eq:UnvStoppingTimeCond}
    \sum_{k=1}^T \log \frac{p_{\tv}(X_k|Y_k)}{q(X_k)} \leq a + \frac{\XY}{2} \log T + \beta + \log \frac{p_{\tv}(X_T|Y_T)}{q(X_T)}
\end{equation}
where we define $a = \log A$. By \eqref{eq:UnvStoppingTimeCond} and Wald's Identity,
\begin{align}
    \E \{T\} &= \frac{\E \left\{\sum_{k=1}^T \log \frac{p_{\tv}(X_k|Y_k)}{q(X_k)}\right\}}{\E\left\{\log \frac{p_{\tv}(X|Y)}{q(X)}\right\}} \\
    &\leq \frac{a + \frac{\XY}{2} \cdot \E\{\log T\} + \beta + C}{C} \label{eq:BoundWithLog}
\end{align}
Since $\log_2 u \leq \frac{u}{v \ln 2} + \log_2 v - \frac{1}{\ln 2}$ for all $u,v>0$, \eqref{eq:BoundWithLog} implies
\begin{equation}
    \E \{T\} \leq \frac{a + \frac{\XY}{2} \left( \log v - \frac{1}{\ln 2} \right) + \beta + C}{C \left( 1 - \frac{\hXY}{C \cdot v\ln 2} \right)}
\end{equation}

For $v = \frac{\log M}{C}$  we obtain
\begin{equation} \label{eq:UnvExpStoppingTime}
    \E \{T\} \leq \frac{a + \frac{\XY}{2} \left( \log \log M - \log C - \frac{1}{\ln 2} \right) + \beta + C}{C \left( 1 - \frac{\hXY}{\log M \ln 2} \right)}
\end{equation}
which corresponds to the following effective rate:
\begin{equation}\label{eq:UnvEffectiveRateParam}
    R = \frac{C \log M \left( 1 - \frac{\hXY}{\log M \ln 2} \right)}{a + \frac{\XY}{2} \left( \log \log M - \log C - \frac{1}{\ln 2} \right) + \beta + C}
\end{equation}

Similarly to the derivation in Chapter \ref{ch:KnownChannel}, we bound the error probability by
\begin{equation*} \label{eq:ErrorProbUnv}
    \Pr \{E_w\} \leq (M-1) \Pr \left\{ \bigcup_{t=1}^{\infty}
    \left\{ \frac{p_U(X^t|Y^t)}{q(X^t)} > A \right\} \right\}
\end{equation*}
where $\{X_k\}_{k=1}^{\infty}$ and $\{Y_k\}_{k=1}^{\infty}$ are independent sequences. Define
\begin{equation}\label{eq:uidef}
    \Phi_t = \begin{cases}
                \frac{p_U(X^t|Y^t)}{p_U(X^{t-1}|Y^{t-1}) \cdot q(X_t)}, \ \prod_{k=1}^{t-1} \Phi_k \leq A \\
                1, \ \text{otherwise}
              \end{cases}
\end{equation}
We can see that
\begin{equation}
    \bigcup_{t=1}^{\infty}
    \left\{ \frac{p_U(X^t|Y^t)}{q(X^t)} > A \right\}
    \Leftrightarrow \prod_{t=1}^{\infty} \Phi_t > A
\end{equation}
Furthermore, we can see that $\E\{\Phi_t\}=1$ for all $t$ since
\begin{equation*}
    \E \left\{ \Phi_t | \prod_{k=1}^{t-1} \Phi_k > A \right\} = 1
\end{equation*}
and
\begin{align}
    \E \left\{ \Phi_t | \prod_{k=1}^{t-1} \Phi_k \leq A \right\}
    &= \E \left\{ \frac{p_U(X^t|Y^t)}{p_U(X^{t-1}|Y^{t-1}) \cdot q(X_t)} \right\} \\
    &= \E \left\{ \E \left\{ \frac{p_U(X^t|Y^t)}{p_U(X^{t-1}|Y^{t-1}) \cdot q(X_t)} | X^{t-1}, Y^t \right\} \right\} \\
    &= \E \left\{ \frac{ \E \left\{ \int_{\Lambda} w(\tv') \frac{p_{\tv'} (x^t|y^t)}{q(x_t)} d\tv' | X^{t-1}, Y^t \right\}}{\int_{\Lambda} w(\tv') p_{\tv'} (x^{t-1}|y^{t-1}) d\tv'} \right\} \\
    &= \E \left\{ \frac{ \sum_{x_t \in \mathcal{X}} q(x_t) \int_{\Lambda} w(\tv') \frac{p_{\tv'} (x^t|y^t)}{q(x_t)} d\tv'}{\int_{\Lambda} w(\tv') p_{\tv'} (x^{t-1}|y^{t-1}) d\tv'} \right\} \\
    &= \E \left\{ \frac{ \int_{\Lambda} w(\tv') \sum_{x_t \in \mathcal{X}} p_{\tv'} (x^t|y^t) d\tv' }{\int_{\Lambda} w(\tv') p_{\tv'} (x^{t-1}|y^{t-1}) d\tv'} \right\} \\
    &= \E \left\{ \frac{ \int_{\Lambda} w(\tv') p_{\tv'} (x^{t-1}|y^{t-1}) d\tv' }{\int_{\Lambda} w(\tv') p_{\tv'} (x^{t-1}|y^{t-1}) d\tv'} \right\} = 1
\end{align}
For an arbitrary $N$ we have
\begin{align}\label{eq:ProdPhi}
    \E \left\{\prod_{t=1}^N \Phi_t \right\} &= \E \left\{ \E \left\{\prod_{t=1}^N \Phi_t | \prod_{t=1}^{N-1} \Phi_t \right\} \right\} \\
    &= \E \{ \Phi_N \} \cdot \E \left\{\prod_{t=1}^{N-1} \Phi_t \right\} \\
    &= \E \left\{\prod_{t=1}^{N-1} \Phi_t \right\} = \ldots = 1
\end{align}
Since the above holds for all $N$, we also have
\begin{equation}\label{eq:InfProdPhi}
    \E \left\{\prod_{t=1}^\infty \Phi_t \right\} = 1
\end{equation}
Thus, similarly to the case of known channel, the error probability can be bounded by
\begin{align}
   \Pr \{E_w\} &\leq (M-1) \Pr \left\{ \bigcup_{t=1}^{\infty}
    \left\{ \frac{p_U(X^t|Y^t)}{q(X^t)} > A \right\} \right\} \label{eq:BoundErrorProbUnknownCh}\\
    &= (M-1) \Pr \left\{ \prod_{t=1}^{\infty} \Phi_t > A \right\} \leq \frac{M-1}{A}
\end{align}
Here again, we choose $A = M/\epsilon$ to obtain $P_e < \epsilon$. Substituting $a = \log A = \log M - \log \epsilon$ into \eqref{eq:UnvEffectiveRateParam} we finally get \eqref{eq:UnvEffectiveRate}.
\end{proof}

\begin{remark}
Interestingly, the upper bound on the error probability in \eqref{eq:BoundErrorProbKnownCh}, obtained when the decoder uses the known channel law $p(x|y)$, applies for an arbitrary probability assignment $p_U(x|y)$, where the only required constraint is that the latter integrates to unity.
\end{remark}

\begin{remark}
As in the case of known channel, we can use randomized decoder here to obtain the following rate:
\begin{equation}\label{eq:UnvEffectiveRateRand}
    R = \frac{C \left( 1 - \frac{\hXY}{\log M \ln 2} \right)}{1 + \frac{C + \beta - \log \delta + \frac{\XY}{2} \left( \log \log M - \log C - \frac{1}{\ln 2} \right)}{\log M}} \cdot \frac{1-\delta}{1-\epsilon}
\end{equation}
for all $0 < \delta < \epsilon$. As we mentioned in Section \ref{sec:CodingThmKnownChannel}, if the required error probability is small, randomization should not be applied. However, if the error probability constraint is loose enough, a better rate may be obtained by optimizing delta in \eqref{eq:UnvEffectiveRateRand}.
\end{remark}

 \section{Discussion}
 \subsection{Comparison to the Known Channel Case}
 Having obtained achievable rates for the cases of both known and unknown channels, it is interesting to compare these results and evaluate the rate degradation due to the unknown channel. For the case of a known channel, the effective rate at \eqref{eq:AchievableRateChannelDec} can be approximated by
 \begin{equation}
    R \thickapprox C \left(1 - \frac{C- \log \epsilon}{\log M} \right)
 \end{equation}

 For the case of unknown channel, we can approximate \eqref{eq:UnvEffectiveRate} by
 \begin{align}
    R_U &\thickapprox C \left(1 - \frac{C- \log \epsilon}{\log M} \right) \nonumber \\
      &- C \left( \frac{\XY}{2} \frac{\log \log M}{\log M} + \frac{(\hXY-1)/\ln 2 + \beta + \log C}{\log M} \right) + O \left( \frac{1}{\log^2 M}\right)
 \end{align}

 Hence, the penalty for lack of channel knowledge amounts to
 \begin{align} \label{eq:RateDegradation}
    R-R_U &= C \left( \frac{\XY}{2} \frac{\log \log M}{\log M} + \frac{(\hXY-1)/\ln 2 + \beta + \log C}{\log M} \right) \\
          &+ O \left( \frac{1}{\log^2 M}\right) \nonumber
 \end{align}

 The leading term in the latter expression behaves as $O(\log \log M / \log M) = O(\log K/K)$, factorized by the product of the cardinalities of input and output of the channel. It is interesting to compare this result with known results from universal source coding, where the \emph{redundancy}\footnote{The excess of the average codeword length above the entropy of the source.} is dominated by the cardinality of the alphabet of the source \cite{UnvPrediction}, and a term that behaves as $O(\log n/n)$, where $n$ is the source length.

 \subsection{Induced Error Exponent}
 Let us now examine Theorem \ref{thm:UnvRate} in light of the previous results. Equation \eqref{eq:UnvEffectiveRate} implies the following error exponent:
 \begin{equation}\label{eq:UnvErrorExp}
    - \frac{\log \epsilon}{\E\{T\}} = C - R - \frac{\XY}{2} \cdot \frac{\log \log M}{\log M} + O\left(\frac{1}{\log M}\right)
 \end{equation}
 As in the case of a known channel, we see that the error exponent is a linear function of the rate, but an additional term of order $O(\log \log M / \log M)$ is added. Here again, we interpret this term as a penalty for the lack of channel knowledge at the receiver. Furthermore, by taking $M \to \infty$, we can also see that \eqref{eq:UnvErrorExp} coincides with \cite[Proposition 1]{Telatar}.

 \subsection{Training and Channel Estimation}
 In many practical applications, communication over an unknown channel is done by means of channel estimation. In this approach, the transmission includes predefined \emph{training} signals, which are known to the receiver and are used to estimate the channel parameters. As an alternative to the universal communication scheme introduced in this chapter, we can use the following method. Prior to any message transmission, the transmitter sends a training sequence, which the receiver uses to estimate the channel. After the training phase, the transmitter sends the message. The receiver uses the \emph{estimated} channel parameters to decode the message, using, for instance, the decoding rule at \eqref{eq:ChannelDecoderLog}. A drawback from this approach is that even after the channel estimation phase, the residual error in the estimated channel parameters will degrade the performance of the decoder. Furthermore, enhancement of the channel estimation accuracy requires long training sequences, which will introduce non-negligible overhead to the transmission time. Clearly, using training will not lead to the convergence rate of \eqref{eq:RateDegradation}.

\chapter{Extensions}
\label{ch:Extensions}
\section{Joint Source-Channel Coding} \label{sec:JointSC}
In the previous chapters we assumed that the messages conveyed over the channel were equiprobable, which is the case if, for instance, the source of information has been compressed and the message $W$ is the output of the source encoder. Assume now, that the messages have arbitrary probabilities ${\pi(1),\ldots,\pi(M)}$. Each message now contains a different amount of information, which would translate into different codeword length at the output of the source encoder. However, in rateless codes the codeword assigned to each message is always infinite, and the actual codeword length is determined by the decoder. (The effective length of the message depends on the decoder's stopping time.) It is therefore tempting to use rateless codes for an uncompressed source and try to achieve good compression rate and reliable communication simultaneously. To simplify matters, we begin by tackling the case of known channel and postpone the analysis for unknown channel to Section \ref{sec:CompleteUnv}. We use the following generalized version of the encoder \eqref{eq:ChannelDecoderLog}.
\begin{equation} \label{eq:JointScDec}
    g_n(y^n)= \begin{cases}
                w, \ z_{w,1}+\ldots+z_{w,n} \geq a_w \\
                0, \ \text{if no such $w$ exists}
              \end{cases}
\end{equation}
where $a_w$ is a threshold that depends on the message $w$, and we define $a_w= \log A_w$. Repeating the derivation for the error probability done in the previous section, we get by Markov's inequality
\begin{equation}
    \Pr\{E_w\} \leq \sum_{w' \neq w} \frac{1}{A_{w'}}
\end{equation}
By choosing
\begin{equation}
    A_w  = \frac{1}{\epsilon \cdot \pi(w)} \qquad \forall w \in \mathcal{W}
\end{equation}
we get a uniform bound on the error probability
\begin{equation}
    \Pr\{E_w\} \leq \epsilon \cdot \sum_{w' \neq w} \pi(w') \leq \epsilon
\end{equation}
which also implies
\begin{equation}
    P_e \leq \epsilon
\end{equation}

Thus, for an appropriate choice of message-dependent threshold values, the average probability of error for the entire message set is bounded by $\epsilon$. Recall, however, that the effective rate depends on the threshold value and therefore needs to be reexamined here. When different thresholds are used for different messages, the stopping time depends on which message crosses the threshold. We can therefore use Wald's equation \eqref{eq:WaldC} conditioned on the true message:
\begin{equation}
    \E\{T|W=w\} \leq \frac{a_w + C}{C}
\end{equation}
where
\begin{equation}
    a_w = \log A_w = - \log \pi(w) - \log \epsilon
\end{equation}
Averaging on the entire message set, we have
\begin{align}
    \E\{T\} &= \E\{\E\{T|W\}\} \leq \frac{\E\{a_{_W}\} + C}{C} \nonumber \\
            &= \frac{\E\{- \log \pi(W)\} - \log \epsilon + C}{C} \nonumber \\
            &= \frac{H(W) - \log \epsilon + C}{C} \label{eq:JointScET}
\end{align}
where $H(W)$ is the entropy rate of the source in bits per symbol. Let us now examine \eqref{eq:JointScET} in a practical setting. Suppose the we wish to convey blocks of $K$ source bits with fixed probability of error $\epsilon > 0$. Since every source symbol contains $\log M$ bits, $K / \log M$ source symbols will be needed. Thus, the rate at which source bits can be conveyed over the channel will be
\begin{align}
    R &= \frac{K}{\E\{T\}} \geq \frac{K \cdot C}{\frac{K \cdot H(W)}{\log M} - \log \epsilon + C} \\
      &= \frac{C}{\mathscr{H}(W) + \frac{C - \log \epsilon}{K}} \\
      &= \frac{C}{\mathscr{H}(W)} \cdot \frac{1}{1 + \frac{C - \log \epsilon}{\mathscr{H}(W) \cdot K}} \\
      &= \frac{C}{\mathscr{H}(W)} \cdot \left(1 - O \left( \frac{1}{K} \right) \right) \label{eq:JointScRateAsym}
\end{align}
where we define $\mathscr{H}(W)=H(W)/ \log M$ as the per-bit entropy of the source.

Note that the encoder used here, as well as the codebook, are the same ones defined in Chapter \ref{ch:DefinitionsAndNotation} and the only change is in the definition of the decoder. The encoder is uninformed on the statistics of the source or the capacity of the channel, yet the rate approaches the optimum rate achievable by an informed encoder. We note the practical implication of such scheme: the compression algorithms can be implemented and maintained at the decoder, while the encoder remains simple and source-independent.

\section{Source Coding with Side Information} \label{sec:SI}
Suppose now, that the source of information emits independent pairs of messages $(W_1,W_2) \in \mathcal{W}_1 \times \mathcal{W}_2$ according to a probability distribution $\pi_{_{W_1,W_2}}(w_1,w_2)$, which are encoded separately and pass through a noiseless channel. Suppose that $R_1$ and $R_2$ are the coding rates of $W_1$ and $W_2$, respectively. By Slepian-Wolf theorem, if $W_1$ is encoded with rate $R_1 \geq H(W_1)$, then $W_2$ can be encoded independently with $R_2 = H(W_2|W_1)$. (This rate pair is a corner point in the achievable rate region.) We will now show that using rateless codes, we can approach this rate with some redundancy due to the usage of finite message set. The encoder of $W_1$ assigns to each message in $\mathcal{W}_1$ an infinite codeword $\mathbf{c}_{w_1} \in \{0,1\}^\infty$, $w_1 = 1,\ldots,|\mathcal{W}_1|$, and transmits it over the channel. The encoder of $W_2$ operates similarly to that of $W_1$ and independently of it, with codewords $\mathbf{d}_{w_2} \in \{0,1\}^\infty$, $w_2 = 1,\ldots,|\mathcal{W}_2|$. The codewords are assumed to be i.i.d.\ Bernoulli$(1/2)$ sequences. To reconstruct $W_1$, the decoder can use the decision rule \eqref{eq:JointScDec}, to to obtain an error probability of
\begin{equation} \label{eq:BoundErrorW1}
    \Pr\{\hat{W_1} \neq W_1\} \leq \frac{\epsilon}{2}
\end{equation}
Since binary code is used and the channel is noiseless, we have $C=1$, so \eqref{eq:JointScET} implies that the expected transmission time for $W_1$ satisfies
\begin{equation} \label{eq:SlepianWolfR1}
    R_1 = \E\{T_1\} \leq \ H(W) - \log \frac{\epsilon}{2} + 1
\end{equation}
Note that the coding rate is defined here as the average codeword length for the message set. Therefore, the effective rate equals the expected transmission time, rather than its reciprocal as in channel coding.

Having decoded message $W_1$, the decoder uses the following decision rule to reconstruct $W_2$:
\begin{equation}
    g^{(2)}_n(y^n,w_1)= \begin{cases}
                w_2, \ z_{w_2,1}+\ldots+z_{w_2,n} \geq a(w_1,w_2) \\
                0,   \ \text{if no such $w_2$ exists}
              \end{cases}
\end{equation}
where
\begin{equation}
    z_{w_2,k} = \log \frac{p(y_k|d_{w_2,k})}{p(y_k)}, \qquad k=1,\ldots,n
\end{equation}
Similar derivation for the error probability as in Section \ref{sec:JointSC} yields
\begin{equation}
    \Pr\{\hat{W_2} \neq w_2 \ | \ W_1 = w_1, W_2 = w_2 \} \leq \sum_{w_2' \neq w_2} \frac{1}{A(w_1,w_2)}
\end{equation}
We choose
\begin{equation}
    A(w_1,w_2) = \frac{1}{\epsilon/2 \cdot \pi_{_{W_2|W_1}}(w_2|w_1)}
\end{equation}
so that
\begin{equation}
    \Pr\{\hat{W_2} \neq w_2 \ | \ W_1 = w_1, W_2 = w_2 \} \leq \epsilon/2 \cdot \sum_{w_2' \neq w_2} \pi_{_{W_2|W_1}}(w_2|w_1) \leq \frac{\epsilon}{2}
\end{equation}
Therefore,
\begin{equation} \label{eq:BoundErrorW2}
    \Pr\{\hat{W_2} \neq W_2\} \leq \frac{\epsilon}{2}
\end{equation}
Using \eqref{eq:BoundErrorW1}, \eqref{eq:BoundErrorW2} and the union bound, we have
\begin{equation}
    \Pr\{\hat{W_1} \neq W_1 \bigcup \hat{W_2} \neq W_2\} \leq \epsilon
\end{equation}

Since $a(w_1,w_2) = - \log \epsilon/2 - \log \pi_{_{W_2|W_1}}(w_2|w_1)$, we can use Wald's equation for the stopping time of decoding $W_2$ to obtain
\begin{align}
    R_2 &= \E\{T_2\} = \E\{\E\{T_2|W_1,W_2\}\} \nonumber \\
        &\leq \E\{a(W_1,W_2)\} + 1 \nonumber \\
        &= \E\{- \log \pi_{_{W_2|W_1}}(W_2|W_1)\} - \log \frac{\epsilon}{2} + 1 \nonumber \\
        &= H(W_2|W_1) - \log \frac{\epsilon}{2} + 1 \label{eq:SlepianWolfR2}
\end{align}

Combining \eqref{eq:SlepianWolfR1} and \eqref{eq:SlepianWolfR2}, we get
\begin{equation} \label{eq:SlepianWolfSumRate}
    R_1 + R_2 = H(W_1,W_2) - 2 \log \frac{\epsilon}{2} + 2
\end{equation}

Similarly to Section \ref{sec:JointSC}, if we take blocks of $K$ source bits and a fixed error probability $\epsilon > 0$, we obtain
\begin{equation}
    R_1 + R_2 = H(W_1,W_2) \cdot \left(1 + O \left( \frac{1}{K} \right) \right)
\end{equation}

\section{Complete Universality} \label{sec:CompleteUnv}

We now consider the case of joint source-channel coding of an unknown source over an unknown channel, with an unknown amount of side-information at the receiver. Initially, we bring together the results of the previous sections to obtain a communication scheme for a source with unknown statistics over an unknown channel. As a straightforward generalization of the universal source coding scheme in Section \ref{sec:RateUnknownChannel}, we use a fusion of the decoders \eqref{eq:ChannelDecoderUnv} and \eqref{eq:JointScDec}, i.e.
\begin{equation} \label{eq:JointScUnvDec}
g_n(y^n)= \begin{cases}
            w, \ p_U(\mathbf{c}_w|y^n) \geq A_w \cdot q(\mathbf{c}_w) \\
            0, \ \text{if no such $w$ exists}
          \end{cases}
\end{equation}
where
\begin{equation}\label{eq:Aw}
    A_w = \frac{1}{\epsilon \cdot \pi(w)}
\end{equation}
Similar derivation to those done at Sections \ref{sec:RateUnknownChannel} and \ref{sec:JointSC} yields the following rate for an uncompressed source $W \in \{1,\ldots,M\}$ over an unknown channel with capacity $C$:
\begin{equation}\label{eq:JointScUnvRate}
    R = \frac{C \left( 1 - \frac{\hXY}{\log M \ln 2} \right)}{\mathscr{H}(W) + \frac{C + \beta - \log \epsilon + \frac{\XY}{2} \left( \log \log M - \log C - \frac{1}{\ln 2} \right)}{\log M}}
\end{equation}
where $\mathscr{H}(W)$ is defined in Section \ref{sec:JointSC}. We note that while the encoder can be ignorant of the source statistic, the decoder needs to know $\pi(w), \ w \in \mathcal{W}$.

We now go one step further and assume that the decoder has no knowledge of the statistics of the source or the channel. Suppose that the source $S$ generates sequences of $L$ symbols from an alphabet $\mathcal{S}$, drawn i.i.d. according to set of $|\mathcal{S}|$ unknown probabilities $\gv$. Each sequence is encoded as one message, hence $M = |\mathcal{S}|^L$. Instead of using the set of thresholds \eqref{eq:Aw}, which depends on the unknown probabilities, we use a universal probability measure \cite{UnvPrediction}
\begin{equation}
    \ph(s^L) = \int u(\gv) \pi_{\gv}(s^L)
\end{equation}
so that
\begin{equation}
    a_w = \log A_w = -\log \epsilon -\log \ph (s^L)
\end{equation}
If the weight function $u(\cdot)$ is chosen to be Jeffreys prior, we get (see \cite[Eq.17]{UnvPrediction})
\begin{equation}
    \E\{a_w\} = -\log \epsilon + H(W) + \frac{|\mathcal{S}|-1}{2} \log \frac{L}{2 \pi e} + O(1)
\end{equation}
Hence, similarly to \eqref{eq:JointScUnvRate} we can achieve the following rate
\begin{equation}\label{eq:CompUnvRate}
    R = \frac{C \left( 1 - \frac{\hXY}{\log M \ln 2} \right)}{\hat{\mathscr{H}}(W) + \frac{C + \beta - \log \epsilon + \frac{\XY}{2} \left( \log \log M - \log C - \frac{1}{\ln 2} \right)}{\log M}}
\end{equation}
where
\begin{equation}
    \hat{\mathscr{H}}(W) = \mathscr{H}(W) + \frac{|\mathcal{S}|-1}{2} \frac{\log L}{\log M} + O\left(\frac{1}{\log M}\right)
\end{equation}
Recall that $L = \log_{|\mathcal{S}|} M$, so
\begin{equation}\label{eq:EmpEntropy}
    \hat{\mathscr{H}}(W) = \mathscr{H}(W) + \frac{|\mathcal{S}|-1}{2} \frac{\log \log M}{\log M} + O\left(\frac{1}{\log M}\right)
\end{equation}
By plugging \eqref{eq:EmpEntropy} into \eqref{eq:CompUnvRate} we get
\begin{equation}\label{eq:CompUnvRateAsym}
    R = \frac{C}{\mathscr{H}(W)} \cdot \left(1 - O \left( \frac{\log K}{K} \right) \right) + O \left( \frac{1}{K} \right)
\end{equation}
where $K = \log M$ is the number of encoded bits. Comparing \eqref{eq:CompUnvRateAsym} to \eqref{eq:JointScRateAsym}, we see that the leading term is unchanged and equals the optimal rate achievable by separated source-channel coding. However, the lack of information affects the rate of convergence, which is now dominated by a $O \left( \frac{\log K}{K} \right)$ term, as opposed to $O \left( \frac{1}{K} \right)$ for an informed decoder.

The implications of the latter result are far-reaching. We have shown that even if the statistics of both the channel and the source are unknown to the decoder, rateless coding not only achieves the best source-channel coding rate as $M \to \infty$, but it also has the same asymptotics of a rateless scheme with an informed decoder. This observation has been made in \cite[Ch.4]{Nadav} for infinitely large message sets. The results obtained here coincide with those of \cite{Nadav}, and also quantify the redundancy caused by the lack of information on the source and the channel, and by the use of finite blocks.

\subsection*{Unknown Side Information at the Decoder}
Similarly to Section \ref{sec:SI}, if the source contains side information $V$ that is known non-causally at the decoder, we can further improve the communication rate. Combining the technique from Section \ref{sec:SI} with the derivation above, we obtain the following rate for universal joint source-channel coding with side information at the decoder:
\begin{equation}\label{eq:CompUnvRateSI}
    R = \frac{C}{\mathscr{H}(W|V)} \cdot \left(1 - O \left( \frac{\log K}{K} \right) \right) + O \left( \frac{1}{K} \right)
\end{equation}
where $\mathscr{H}(W|V)$ is the conditional entropy of the source $W$ given the side information $V$, normalized by $\log M$. Since $\mathscr{H}(W|V) \leq \mathscr{H}(W)$, the side information improves the rate, even if the encoder is uninformed on the amount (or the existence) of the side information.

\chapter{Summary}
\label{ch:Summary}
In this study we developed and analyzed several communication schemes that are all based on the concept of \emph{rateless codes}. In rateless codes, each codeword has an infinite length and the decoding length is dynamically determined by the confidence level of the decoder. Throughout this study, we allowed the coding schemes to have a fixed error probability, while aiming to achieve shortest mean transmission time, or equivalently, the highest rate. This approach is different than the prevalent one, in which the communication rate is held fixed and the codebook is enlarged indefinitely so that the error probability vanishes. We demonstrated how rateless codes, combined with sequential decoding, can be used in basic communication scenarios such as communication over a DMC, but can also be used to solve more complex problems, such as communication over an unknown channel. The decoding methods introduced here enabled us to obtain results for finite message set, while previous studies were restricted to asymptotic results.

We began by describing rateless codes and surveyed some previous results related to such coding schemes. Then, we introduced the sequential decoder that uses a known channel law. Using Wald's theory and the notion of stopping time, we obtained an upper bound for the mean transmission time for a fixed error probability, and the resulting effective rate is shown to approach to the capacity of the channel as the size of the message set, $M$, grows. We also obtained an upper bound for the rate for a fixed error probability. The upper bound is not tight for small $M$, but it converges to the achievable rate as $M \to \infty$. We conjecture that a stronger converse can be found, which will be tighter also in the non-asymptotic realm. Although we developed the above-mentioned scheme for a DMC, we also demonstrated that it is applicable in a memoryless Gaussian channel.

For the case of an unknown channel we introduced a novel decoding metric. Unlike previous studies, the universal decoding metric in not based on empirical mutual information, but on a mixture probability assignment. For an appropriate choice of mixture, we were able to bound the difference between the universal metric and the one used by an informed decoder. Thus, we used the results obtained for an informed decoder to upper bound the mean transmission time in the universal case.

We then applied rateless coding to more advanced scenarios. We showed how with only a minor change in the sequential decoder, we can easily use rateless codes as a joint source-channel coding scheme. We also used rateless coding for source coding with side information, obtaining the optimum Slepian-Wolf rate for this setting. Finally, we combined the techniques for universal channel coding, joint source-channel coding and source coding with side information and demonstrated that even without any information on the source, the channel or the amount (or even the existence) of side information---reliable communication is feasible, and the rate can be analyzed even for a finite message set.

\bibliographystyle{IEEETran}
\bibliography{Bibliography}
\nocite{*}

\end{document}